\documentclass[12pt,a4paper]{article}

\usepackage{times}
\usepackage{amsthm,amsfonts,amssymb,amsmath}

\usepackage[T1]{fontenc}
\usepackage[cp1250]{inputenc}
\newcommand{\be}{\begin{equation}}
\newcommand{\ee}{\end{equation}}

\newcommand{\ba}{\begin{eqnarray}}
\newcommand{\ea}{\end{eqnarray}}

\usepackage{cite}
\usepackage{color}
\newtheorem{thm}{Theorem}[section]

\newcommand{\p}{\partial}

\input amssym.def
\newsymbol\square 1003

\begin{document}
\centerline{\bf  Covariant description of isothermic surfaces}

\bigskip

-\centerline{J. Tafel}

\noindent
\centerline{Institute of Theoretical Physics, University of Warsaw,}
\centerline{Ho\.za 69, 00-681 Warsaw, Poland, email: tafel@fuw.edu.pl}

\bigskip

\begin{abstract}
We present a covariant formulation of the Gauss-Weingarten equations and the Gauss-Mainardi-Codazzi equations for   surfaces in 3-dimensional curved spaces. We derive a coordinate invariant condition on the first and second fundamental form which is locally  necessary and sufficient for the surface to be isothermic. We show how to construct isothermic coordinates.
\end{abstract}

\noindent
Keywords: isothermic surfaces, the Gauss-Weingarten equations

%PACS: 02.30.Ik, 02.40.Hw

\noindent
MCS: 53C42, 53A35

\section{Introduction} 
Studies of isothermic surfaces were originated in 1837  by  Lame \cite{l} who considered surfaces of constant temperature in a solid body. Bertrand \cite{b} is an author of  the  present definition of  isothermic  surfaces as  surfaces  admitting  coordinates in which the first fundamental form is conformally flat and the second one is diagonal. Their basic properties were found by Bour \cite{bo}, Darboux \cite{d}, Calapso \cite{c} and Bianchi \cite{bi} (see \cite{k} for a review). A transformation given by Darboux and developed by Bianchi allow to obtain a family of new isothermic surfaces from a given one. This property was crucial for an interpretation of isothermic surfaces as soliton surfaces by Cie\'{s}li\'{n}ski, Goldstein and Sym \cite{cgs}. They recognized the Gauss-Mainardi-Codazzi equations for the isothermic surfaces  as a completely integrable system. These results initiated a new interest in isothermic surfaces (see \cite{rs,h,bc} for a review), including higher dimensions \cite{bpt,b1},  
discretizations \cite{bp}  and  their generalizations (see e.g. \cite{d1}).

Most studies on isothermic surfaces refer to surfaces in a flat or conformally flat 3-dimensional space (see \cite{bdpp} and references therein for surfaces in symmetric spaces). Moreover, their  definition is coordinate dependent. The main result of this paper is   a coordinate invariant characteristic of isothermic surfaces which is also valid in a general curved 3-dimensional space.

In section 2 we give a covariant formulation of the generalized  Gauss-Weingarten (GW) equations  and their integrability conditions known as the Gauss-Mainardi-Codazzi (GMC) equations. We also present a method of a construction of a surface in flat space $R^3$ with prescribed fundamental forms, which is simpler than passing through the full system of the GW equations.

In section 3 we analyze equations for a coordinate transformation which put a 2-dimensional Riemannian metric $g_{ab}$ into a conformally flat  form   and, simultaneously, diagonalizes a symmetric tensor $K_{ab}$. Their integrability condition reads $d\omega=0$, where $\omega$ is a differential form defined by the metric and the tensor.  This equation, with $g_{ab}$ and $K_{ab}$ corresponding to fundamental forms $g_I$ and $g_{II}$ of a surface, should complete the GW and the GMC equations in the case of isothermic surfaces. It is conformally invariant and it allows an easy verification, in any system of coordinates, if a surface corresponding to some $g_I$ and $g_{II}$ is isothermic. If it is satisfied a construction of isothermic coordinates can be performed in quadratures. 

If we represent a surface as a graph of the form $x^3=f(x^1,x^2)$ and express  $g_I$ and $g_{II}$ in terms of $f$, then  condition $d\omega=0$ becomes a fourth order nonlinear equation for the function $f$. In principle, this  equation  can be  used  to find isothermic surfaces in a Riemannan 3-space.

\section{The Gauss-Weingarten equations}
Let $S$ with coordinates $\xi^a$, where $a\in\{1,2\}$, be  a 2-dimensional surface in a 3-dimensional manifold $M$ with  metric $\tilde g$ and coordinates $x^i$, $i\in\{1,2,3\}$. The induced metric (the first fundamental form)     is given by
\be\label{2}
g_I=g_{ab}d\xi^a d\xi^b\ ,\ \ g_{ab}=\tilde g_{ij}x^i_{,a}x^j_{,b}\ ,
\ee 
where
\be\label{1}
\frac{\p x^i}{\p \xi^a}=x^i_{,a}\ .
\ee
Vectors $t_a=x^i_{,a}\p_i$ are tangent to $S$ and orthogonal to the unit  vector $n=n^i\p_i$ 
\be\label{3}
n_ix^i_{,a}=0\ ,\ \ n_in^i=1\ .
\ee

The field $x^i_{,a}$ transforms as a vector under transformations of coordinates $x^i$ and as a covector under a change of $\xi^a$. Thus, it makes sense to define its covariant derivative on $S$ as
\be\label{3a}
Dx^i_{,a}=dx^i_{,a}+\tilde \omega^i_{\ j}x^j_{,a}-\omega^b_{\ a}x^i_{,b}\ ,
\ee
where $\tilde \omega^i_{\ j}$ and $\omega^a_{\ b}$ are the Levi-Civita connections related, respectively, to metric $\tilde g$ and $g_I$.
The tensorial 1-form   $Dx^i_{,a}$ transforms as $x^i_{,a}$ under a change of coordinates. Its decomposition in the basis $d\xi^a$   yields 
\be\label{4}
Dx^i_{,a}=x^i_{\ ab}d\xi^b\ ,
\ee
where
\be\label{5a}
x^i_{\ ab}=x^i_{,ab}+\tilde \Gamma^i_{\ jk}x^j_{,a}x^k_{,b}-\Gamma^c_{\ ab}x^i_{,c}
\ee
and $\tilde\Gamma^i_{\ jk}$ and $\Gamma^c_{\ ab}$ are, respectively, the Christoffel symbols of $\tilde g$ and $g_I$. Note that $x^i_{\ ab}$ is symmetric in the lower indices. 

For any pair $(a,b)$  we can decompose vector $x^i_{\ ab}\p_i$ in the basis  $\{t_a,n\}$
\be\label{6a}
x^i_{\ ab}=x^i_{,d}f^d_{\ ab}+n^if_{ab}\ .
\ee
Unknown coefficients $f_{ab}$ and $f^c_{\ ab}$ have to be symmetric in indices $(a,b)$. The contraction of (\ref{6a}) with $n_i$ yields $f_{ab}=-K_{ab}$, where $K_{ab}$ are components of  the second fundamental form (the exterior curvature) of $S$  given by
\be\label{7a}
g_{II}=K_{ab}d\xi^ad\xi^b=Dn_idx^i=(n_{i,j}-\tilde \Gamma^k_{\ ij}n_k)x^j_{,a}x^i_{,b}d\xi^ad\xi^b\ .
\ee
 A contraction of (\ref{6a}) with $\tilde g_{ij}x^j_{,c}$ shows that $f_{abc}=-f_{bac}$. This together with $f_{abc}=f_{acb}$ yields $f_{abc}=0$. Hence,
\be\label{8a}
Dx^i_{,a}=-n^iK_{ab}d\xi^b.
\ee

Equation (\ref{8a}) is a part of the  Gauss-Weingarten  equations \cite{h}. In order to complete them we calculate the covariant derivative of $n^i$ written in the form 
\be\label{9a}
n^i=\frac 12\tilde \eta^i_{\ jk}x^j_{,a}x^k_{,b}\eta^{ab}.
\ee
Here $\tilde\eta^{ijk}$ and $\eta^{ab}$ are the Levi-Civita completely antisymmetric tensors corresponding to $\tilde g$ and $g_I$, respectively. Taking into account $(\ref{8a})$ and $D\eta^{ijk}=D\eta^{ab}=0$ one obtains
\be\label{10a}
Dn^i=x^i_{,a}K^a_{\ b}d\xi^b\ .
\ee
Equations (\ref{8a}) and (\ref{10a}) generalize the GW equations. Note that for a nonvanishing Gauss curvature $K=\det{K^b_{\ a}}\neq 0$  condition (\ref{10a}) together with $n^2=1$, $K_{[ab]}=0$ and $g_{ab}=x^i_{,a}\tilde g_{ij}x^j_{,b}$ implies $n_ix^i_{,a}=0$ and (\ref{8a}). To show this one can take contractions of (\ref{10a}) with $n_i$ and $g_{ij}x^j_{,c}$ and repeat the proof of $f_{abc}=0$ following (\ref{6a}).

Assume that a 3-dimensional metric $\tilde g_{ij}$ on $M$ and tensors $g_{ab}$ and $K_{ab}$ on a 2-dimensional manifold $S$ are given. 
If there exist solutions $x^i(\xi^a)$ and $n^i(\xi^a)$ of equations  (\ref{8a}), (\ref{10a}) and constraints (\ref{2}) and (\ref{3}), then   $S$ can be immersed in $M$ in such a way  that $g_{ab}$ and $K_{ab}$ are, respectively, the first and the second fundamental form of $S$. Let us consider  equations (\ref{1}), (\ref{8a}) and (\ref{10a}) as an overdetermined  system for $x^i$, $x^i_{,a}$ and $n^i$.  Covariant derivatives of constraints (\ref{2}) and (\ref{3}) vanish, hence, it is sufficient to assume them in a single point.  Integrability of (\ref{1}) is assured by the symmetry $x^i_{\ ab}=x^i_{\ ba}$ following from (\ref{8a}). Integrability conditions of (\ref{8a}) and (\ref{10a}) yield the generalized GMC equations. They coincide with the following equations 
obtained within the Arnowitt-Deser-Misner   decomposition \cite{adm}  of the Einstein tensor of $\tilde g$
\be\label{11a}
\tilde G_{ij}n^in^j=\frac 12(H^2-K_{ab}K^{ab}-R)
\ee
\be\label{12a}
\tilde G_{ij}n^ix^j_{,a}=(H\delta^b_{\ a}-K^b_{\ a})_{|b}\ .
\ee
Here $\tilde G_{ij}$ is the Einstein tensor of $\tilde g$, $H=K^a_{\ a}$ is the mean curvature of $S$, $R$ is the Ricci scalar of $g_I$ and ${}_{|b}$ denotes the covariant derivative corresponding to $g_I$. Note that for a 2-dimensional surface $\frac 12(H^2-K_{ab}K^{ab})$ is just the Gaussian curvature $K$.

For general $\tilde g$ equations (\ref{11a}) and (\ref{12a}) are additional constraints for $x^i$, $x^i_{,a}$ and $n^i$. They can be further differentiated until one obtains conditions which contain only components of fundamental forms $g_I$ and $g_{II}$.   
Equations (\ref{11a}) and (\ref{12a}) immediately reduce to equations for $g_I$ and $g_{II}$ if $\tilde g$ is an Einstein's metric,  $\tilde G_{ij}=\Lambda\tilde g_{ij}$. 3-dimensional metrics of this type are locally equivalent to  standard metrics of the  flat space $R^3$, sphere $S_3$ or pseudosphere (hyperbolic space) $H_3$. In these cases, given forms $g_I$ and $g_{II}$ which satisfy  (\ref{11a}) and (\ref{12a}), one can find a corresponding surface via solving the linear system (\ref{8a})-(\ref{10a}).

In the case of $R^3$ and $K\neq 0$ this procedure can be  simplified by using the tensor 
\be\label{12c}
g_{III}=K_{ac}K^c_{\ b}d\xi^ad\xi^b=Hg_{II}-Kg_I
\ee
called the third fundamental form.
It follows  from (\ref{10a})  that
\be\label{12b}
g_{III}=\bar n_{,a}\bar n_{,b}d\xi^ad\xi^b\ ,
\ee
where $\bar n$ denotes the normal vector  in the Cartesian coordinates.
In the region, where $n^3\neq 1$, equation (\ref{12b}) can be written in the form
\be\label{12d}
g_{III}=\frac{dn dn^*}{(1+\frac 14nn^*)^2}\ ,
\ee
where $n=\frac{n^1+in^2}{1-n^3}$ and $n^*$ is its complex conjugate. If $K\neq 0$  then  $g_{III}$ is nondegenerate, functions $n$ and $n^*$ must be independent and  the rhs of (\ref{12d}) is the standard spherical metric. 
An interesting consequence of (\ref{12c}) and (\ref{12b}) is that metric of any minimal ($H=0$) surface in $R^3$  satisfies 
\be
g_I=|K|^{-1}d\bar nd\bar n\ ,
\ee
exactly as it is for constant radius spheres, which are not minimal surfaces.

Given  $g_I$ and $g_{II}$ satisfying the GMC equations, from (\ref{12d}) one can derive the complex stereographic coordinate $n$ and normal vector $\bar n$. To this end one should first find  complex coordinates $\xi$, $\xi^*$ in which form (\ref{12c}) is conformally flat
\be
g_{III}=e^{2u}d\xi d\xi^*\ .\label{12e}
\ee
This can be done by solving the complex linear equation
\be\label{12p}
d\xi\wedge (\theta^1+i\theta^2)=0\ ,
\ee
where $\theta^a$ is an orthonormal basis for tensor $g_{III}$. Equation (\ref{12p}) has the form $v(\xi)=0$, where $v$ is a complex vector field. If it is analytic, equation (\ref{12p}) can be solved by the method of characteristics.

It follows from (\ref{12d}) and (\ref{12e})  that $n$ is a holomorphic function of $\xi$ (or antiholomorphic but then we can change this by means of  a reflection in $R^3$) and 
\be
e^{-u}=\frac{1}{\sqrt{\dot n^*}}\frac{1}{\sqrt{\dot n}}+\frac{n^*}{\sqrt{\dot n^*}}\frac{n}{4\sqrt{\dot n}}\ ,\label{12f}
\ee
where $\dot n=n_{,\xi}$. Let us consider  holomorphic extension of equation (\ref{12f}) to two independent complex variables, still denoted by $\xi$ and $\xi^*$, and then fix $\xi^*=\xi_0^*$. It follows that
\be
e^{-u_0}=A_{11}\frac{1}{\sqrt{\dot n}}+A_{21}\frac{n}{4\sqrt{\dot n}}\ ,\label{12h}
\ee
where $u_0(\xi)=u(\xi,\xi_0^*)$ and $A_{11}$ and $A_{21}$ are complex numbers.  
If we apply the same procedure to  the derivative of (\ref{12f}) over $\xi^*$  we obtain
\be
(e^{-u_0})_{,\xi_0^*}=A_{12}\frac{1}{\sqrt{\dot n}}+A_{22}\frac{n}{4\sqrt{\dot n}}\ .\label{12i}
\ee
Matrix $(A_{ab})$ is an element of the group $SL(2,C)$. Equations (\ref{12h}) and (\ref{12i}) can be easily solved with respect to $n$ and $\dot n$ giving
\be\label{12l}
n=-4\frac{A_{11}u_{,\xi^*_0}+A_{12}}{A_{21}u_{,\xi^*_0}+A_{22}}\ ,\ \ \dot n=\frac{e^{2u_0}}{(A_{21}u_{,\xi^*_0}+A_{22})^{2}}\ .
\ee

Equation (\ref{12f}) cannot define the stereographic coordinate $n$  better than up to the following transformations, with complex parameters $\alpha$ and $\beta$,
\be\label{12m}
n\rightarrow\frac{\alpha n+\beta}{-\beta^* n+\alpha^*}\ ,\ \ \ |\alpha|^2+|\beta|^2=1
\ee
which correspond to rigid rotations in $R^3$.
One can  use (\ref{12m}) to reduce (\ref{12l}) to
\be
n=-4A(u_{,\xi^*_0}+B)\ ,\ \ \dot n= A e^{2u_0}\label{12j}
\ee
with a complex constant $B$ and a real positive constant  $A$.  Substituting  (\ref{12j}) back into (\ref{12f}) yields equality
\be
1+4A^2|u_{,\xi^*_0}+B|^2=Ae^{u_0+u_0^*-u}\ ,\label{12k}
\ee
which should be satisfied for all values of $\xi$. In order to fit constants $A$ and $B$ algebraically it is sufficient to impose (\ref{12k}) at three points. Given them we can calculate the normal vector $\bar n$ and substitute it into (\ref{10a}) written as
\be\label{12n}
x^i_{,a}=K^{-1}(H\delta_a^{\ b}-K_a^{\ b})n^i_{,b}\ .
\ee
For complex conjugated $\xi$ and $\xi^*$ functions $x^i(\xi,\xi^*)$ follow from  integration of (\ref{12n}) up to constants, which together with (\ref{12m}) represent an Euclidean motion in $R^3$. In this way one obtains a parametric description of the surface with prescribed  fundamental forms, provided that they satisfy the GMC equations. This procedure seems to be much simpler than passing through the full system of the GW equations (see, e.g. \cite{kta}).
\section{Isothermic surfaces}\
The conformal transformation 
\be\label{13a}
\tilde g\rightarrow \Omega^2\tilde g
\ee
implies 
\be\label{21c}
g_I\rightarrow \Omega^2 g_I\ ,\ \ g_{II}\rightarrow \Omega g_{II}+n(\Omega)g_I\ .
\ee
Fundamental properties of  isothermic surfaces, which have conformally flat $g_I$ and diagonal $g_{II}$ in some coordinates $\xi'^a$,  are preserved by these transformations. In view of this, if $\tilde g$ is conformally flat, then in order to find isothermic surfaces  one can first look for solutions of  the flat GMC equations written in  the  coordinates $\xi'^a$.  In this case one can reduce the problem to the Calapso equation \cite{rs}. This approach fails if $\tilde g$ is not conformally flat. For this reason, in this section we will characterize isothermic surfaces independly of the GW and GMC equations. Our aim is to find a coordinate invariant condition which is locally necessary and sufficient for the existence of isothermic coordinates.  An interpretation of $K_{ab}$ as the second fundamental form is not relevant for this condition.

\begin{thm}
 Let $g_{ab}$ and $K_{ab}$ be, respectively,  a 2-dimensional Riemannian metric and a symmetric tensor, which are locally linearly independent. Then, $g_{ab}$ is conformally flat and $K_{ab}$ is diagonal in some coordinates iff
\be
d\omega=0\ ,\label{22b}
\ee
where
\be
\omega=k_a^{\ b}k_{b\ |c}^{\ c}d\xi^a\ ,\label{22}
\ee
\be 
k_{ab}=\alpha^{-1} (K_{ab}-\frac 12K_c^{\ c}g_{ab})\label{22a}
\ee
 and $\alpha>0$ is a normalization factor such that  $k_{ab}k^{ab}=2$. 
\end{thm}
\begin{proof}
Let $\xi'^a$ be  coordinates in which $g'_{12}=0$, $g'_{11}=g'_{22}$ and $K'_{12}=0$.   A transformation from general coordinates $\xi^a$ to $\xi'^a$ has to satisfy the conditions
\be
L^a_{\ 1}g_{ab}L^b_{\ 2}=0\label{5}
\ee
\be
L^a_{\ 1}g_{ab}L^{b}_{\ 1}=L^a_{\ 2}g_{ab}L^{b}_{\ 2}\label{6}
\ee
\be
L^a_{\ 1}K_{ab}L^{b}_{\ 2}=0\ ,\label{7}
\ee
where $L^a_{\ b}=\frac{\p \xi^a}{\p \xi'^b}$. 
In view of (\ref{5}) equation (\ref{7}) can be replaced by 
\be
L^a_{\ 1}\hat K_{ab}L^{b}_{\ 2}=0\ ,\label{7a}
\ee
where 
\be
\hat K_{ab}=K_{ab}-\frac 12K_c^{\ c}g_{ab}
\ee
is the traceless part of $K_{ab}$.

Let $-\alpha$ be a point dependent eigenvalue of $\hat K^a_{\ b}$. It must be real since any 2-dimensional metric $g_{ab}$ admits  conformally flat coordinates  and in these coordinates matrix $\hat K^a_{\ b}$ is  symmetric. Since $\det{(\hat K_{ab}+\alpha g_{ab})}=0$ and $\hat K_{ab}$ is traceless  it can be written in the form
\be
\hat K_{ab}=\alpha  (2u_au_b-g_{ab})\ ,\label{8}
\ee
where $u^au_a=1$. Note that $\alpha\neq 0$ since  $K_{ab}$ is not proportional to $g_{ab}$. Thus, locally there is $\alpha> 0$ or $\alpha< 0$ depending on our choice of the eigenvector $u$ (note that  transformation $u^a\rightarrow\eta^{ab}u_b$, $\alpha\rightarrow -\alpha$ preserves (\ref{8})). 
By virtue of (\ref{8}) equation (\ref{7}) leads to $L^a_{\ 1}u_a=0$ (or $L^a_{\ 2}u_a=0$ but then we can interchange indices 1 and 2). Hence
\be
L^a_{\ 1}=\gamma \eta^{ab}u_b\ ,\label{9}
\ee
where $\gamma$ is a function. Substituting (\ref{9}) into (\ref{5}) yields 
$L^a_{\ 2}=\gamma' u^a$. Equation (\ref{6}) shows that $\gamma'^2=\gamma^2$. Thus, equations (\ref{5})-(\ref{7}) are equivalent to (\ref{9}) and 
\be
L^a_{\ 2}=\pm \gamma u^a\ .\label{11}
\ee

From  (\ref{9}) and (\ref{11}) we obtain the following expressions for the matrix $L^{-1}$
\be
\frac{\p \xi'^1}{\p \xi^a}=\lambda \eta_{ab}u^b\label{12}
\ee
\be
\frac{\p \xi'^2}{\p \xi^a}=\pm \lambda u_a\ ,\label{13}
\ee
where $\lambda$ is a nonvanishing function. Conditions  (\ref{12}) and (\ref{13}), considered as equations for coordinates $\xi'^a$, are integrable provided
\be
(\ln{|\lambda|})_{,a}=u^bu_{a|b}-u^b_{\ |b}u_a\ .\label{14}
\ee
Equation  (\ref{14}) can be written in the form
\be
(\ln{|\lambda|})_{,a}=-\frac 12k_a^{\ b}k_{b\ |c}^{\ c}\ ,\label{14a}
\ee
where
\be
k_{ab}=2u_au_b-g_{ab}\ .\label{21}
\ee
It follows from (\ref{21}) and $u^2=1$ that $k_{ab}k^{ab}=2$. Equality (\ref{8}) implies  (\ref{22a}) and 
\be
 \alpha^2=\frac 12\hat K_{ab}\hat K^{ab}\ .\label{21a}
\ee
Equation (\ref{22b}) is just an integrability condition of (\ref{14a}). 
\end{proof}
If $g_{ab}$ and $K_{ab}$ are to be fundamental forms of a surface $S$ in a 3-dimensional Riemannian manifold $M$ they should be compatible with the GMC equations.  Points, where $\hat K_{ab}=0$ are called umbilics. Out of them the normalized tensor $k_{ab}$ can be defined (it can be still prolongated to umbilics, in which $\hat K_{ab}$ has a zero of a finite order). It follows from (\ref{21c}) that $k_a^{\ b}$ and $d\omega$ are conformally invariant. Equation (\ref{22b}) is  an additional condition for the surface to be isothermic. In contrary to other approaches, for instance via the Hopf differential \cite{bob}, this condition  can be verified in any system of coordinates. 

The proof of Theorem 1 indicates how to find isothermic coordinates $\xi'^a$ if condition (\ref{22b}) is satisfied.
We should first find $\lambda$ from (\ref{14a}) and an unit vector $u^a$ from
\be
k^a_{\ b}u^b=u^a\ .\label{21b}
\ee
Then  equations (\ref{12}) and (\ref{13}) can be solved in quadratures. They are integrable thanks to (\ref{22b}). This procedure is much simpler from that presented in \cite{atk}.

It follows from (\ref{21}) that $k^a_{\ b}$ is an involutive matrix
\be\label{23}
k^a_{\ b}k^b_{\ c}=\delta^a_{\ c}\ .
\ee
Substituting (\ref{22a}) and (\ref{23}) into equation (\ref{22b}) yields
\be
\eta^{ad}(\alpha^{-2}\hat K_a^{\ b}\hat K_{b\ |c}^{\ c})_{|d}=0\ ,\label{25}
\ee
Equation (\ref{25}) is automatically satisfied if   $\tilde  G_{ij}=\Lambda \tilde g_{ij}$  and $H=const$, since then equation (\ref{12a}) yields $\hat K_{b\ |c}^{\ c}=0$. This reflects a known fact that constant mean curvature surfaces in space forms $R^3$, $S_3$ or $H_3$ are  isothermic.

\section{Summary}
We have been considering surfaces $S$ in a three dimensional curved space $M$ with metric $\tilde g$.  The Riemannian connections on $M$ and $S$ allow to write    the GW equations  and the GMC equations in a covariant way (equations (\ref{8a})-(\ref{12a})). Using the third fundamental form we have shown how to simplify  a construction of a surface in the flat space $R^3$ given fundamental forms $g_I$ and $g_{II}$ satisfying the GMC equations.

Starting with a 2-dimensional Riemannian metric $g_{ab}$ and a symmetric tensor $K_{ab}$ we defined a 1-form $\omega$, which must be closed if there exist coordinates in which metric is conformally flat and $K_{ab}$ is diagonal (Theorem 1). If these tensors correspond to fundamental forms of a surface $S$, condition $d\omega=0$ is locally necessary and sufficient for the surface to be isothermic.
 This condition   is preserved by transformations of coordinates  and   conformal transformation (\ref{13a}). This is an advantage with respect to standard  descriptions of isothermic surfaces.  If $d\omega=0$, a construction of isothermic coordinates can be easily performed.
 
 In principle, one can also use  equation $d\omega=0$  in order to find isothermic surfaces together with $g_I$ and $g_{II}$. Locally, every surface $S\subset M$ is a graph of a function $f$ of two coordinates.
 The induced metric 
and the normal vector of $S$  depend on $f$ and its first derivatives. Hence, the exterior curvature tensor  depends  on second derivatives of $f$ and equation $d\omega=0$ is a fourth order nonlinear equation for $f$. Unfortunately,  this equation is rather complicated even in the flat space $R^3$, compared e.g. to the Calapso equation.

\end{document}